\documentclass[pdflatex,sn-mathphys-num]{sn-jnl}

\usepackage{graphicx}%
\usepackage{multirow}%
\usepackage{amsmath,amssymb,amsfonts}%
\usepackage{amsthm}%
\usepackage[title]{appendix}%
\usepackage{xcolor}%
\usepackage{textcomp}%
\usepackage{manyfoot}%
\usepackage{booktabs}%
\usepackage{algorithm}%
\usepackage{algorithmicx}%
\usepackage{algpseudocode}%
\usepackage{listings}%

\theoremstyle{thmstyleone}%
\newtheorem{theorem}{Theorem}
\theoremstyle{thmstyletwo}%

\theoremstyle{thmstylethree}%
\newtheorem{definition}{Definition}%

\usepackage{physics}
\usepackage{qcircuit}
\usepackage{tikz}
\usepackage{array}
\usepackage{longtable}

\newcommand{\hu}{\hat{U}}

\newcommand{\dagg}[1]{#1^{\dagger}}

\newtheorem{lemma}{Lemma}

\begin{document}

\title[Article Title]{Unconditionally successful quantum Time-Marching algorithm via LCU for nonlinear Burgers' equation}


\author*[1,2]{\fnm{Niccol\'o} \sur{Fonio}}\email{niccol.fonio11@gmail.com}

\author[2,3]{\fnm{Giuseppe} \sur{Di Molfetta}}\email{giuseppe.dimolfetta@lis-lab.fr}

\author[1]{\fnm{Pierre} \sur{Sagaut}}\email{pierre.sagaut@univ-amu.fr}

\affil[1]{\orgdiv{Aix-Marseille Université}, \orgname{CNRS, Centrale Méditerranée, M2P2}, \orgaddress{\city{Marseille}, \country{France}}}

\affil[2]{\orgdiv{Aix-Marseille Université}, \orgname{CNRS, LIS}, \orgaddress{\city{Marseille}, \country{France}}}

\affil[3]{\orgname{Institute universitaire de France}, \orgaddress{\city{Paris}, \country{France}}}



\abstract{Most recently proposed quantum algorithms for solving linear and nonlinear partial differential equations rely on non-unitary operations. These operations are typically implemented probabilistically, requiring postselection and thus increasing the computational cost. We show that quantum lattice gas algorithms enable unconditionally successful quantum simulation of nonlinearities, yielding, to our knowledge, the first quantum algorithm for Burgers' equation whose time steps can be concatenated without probabilistic failure. The key idea is to exploit the correspondence between the stochasticity of quantum measurement in the linear combination of unitaries framework and the intrinsic randomness of the classical lattice gas algorithm. In doing so, we identify general properties that characterize probabilistic classical algorithms amenable to this time-marching formulation, and illustrate the approach with an additional application.}





\maketitle

\section{Introduction}\label{sec1}
Quantum algorithms for partial differential equations (PDEs) have attracted considerable attention because of their potential computational advantages in scientific computing. Significant progress has been achieved for linear PDEs through techniques based on Hamiltonian simulation, linear combination of unitaries (LCU), and, more recently, quantum singular value transformation (QSVT) \cite{harrow2009quantum,childs2012hamiltonian,berry2014exponential,low2019hamiltonian,gilyen2019quantum}. These developments have also motivated quantum algorithms for nonlinear PDEs by combining efficient linear solvers with linearization techniques such as Carleman linearization, Koopman--von Neumann linearization, homotopy analysis and Schrödingerization \cite{liu2021efficient,jin2022quantum,hu2024quantum,novikau2025quantum,gan2025provably,jennings2025quantum,bharadwaj2025compact,bharadwaj2025quantum}. Nevertheless, the additional approximation introduced by the linearization itself may ultimately limit the effectiveness of these approaches for nonlinear dynamics \cite{lewis2024limitations, jennings2025end}. Consequently, alternative formulations that simulate nonlinear equations directly remain an active research direction, including quantum lattice gas algorithms and related kinetic approaches \cite{budinski2021quantum,schalkers2024importance,wang2025quantum, zamora2025efficient, fonio2026two}.

Among the different paradigms for quantum PDE solvers, time-marching (TM) algorithms constitute a particularly attractive class. Rather than constructing the solution globally, they reproduce the classical numerical scheme by repeatedly applying an evolution operator corresponding to a single time step. For nonlinear equations, this evolution is generally non-unitary, making its efficient quantum implementation one of the central challenges in the design of TM algorithms.

Non-unitary operators are commonly implemented through block encoding, QSVT, or related LCU-based techniques. Although these methods provide powerful algorithmic tools, they generally implement the desired evolution only probabilistically, requiring postselection or amplitude amplification. When the same non-unitary operator must be applied over many time steps, the overall success probability decreases exponentially unless additional amplification procedures are introduced, increasing both circuit depth and resource requirements. Several strategies have been proposed to mitigate this issue, including oblivious amplitude amplification \cite{zecchi2025improved}, uniform singular value amplification \cite{fang2023time}, block encoding of the complete time evolution \cite{over2025quantum,he2026time}, Hamiltonian-simulation approaches \cite{brearley2024quantum}, and recent time-marching formulations for linear lattice Boltzmann methods \cite{wawrzyniak2025linearized,bediche2025fully,fonio2025adaptive}. However, these approaches either remain restricted to linear dynamics or still rely on probabilistic implementations of the non-unitary evolution.

To the best of our knowledge, no quantum time-marching algorithm has yet demonstrated unconditional concatenation of successive nonlinear time steps. Establishing such a formulation would remove one of the principal bottlenecks of repeated non-unitary evolution and considerably broaden the applicability of quantum time-marching methods.

In this work, we address this question by constructing a quantum time-marching algorithm for Burgers' equation based on a quantum lattice gas cellular automaton. The key observation is that the intrinsic stochasticity of the classical lattice gas algorithm naturally matches the probabilistic character of quantum measurement within the LCU framework, allowing successive time steps to be concatenated without loss of success probability. More generally, we derive the conditions under which this correspondence can be exploited and formulate design principles for identifying other probabilistic classical algorithms amenable to the same approach. As an illustration of these principles, we additionally prove that a naive randomly sampled finite-difference discretization of the advection equation cannot satisfy these conditions with a single ancilla under amplitude encoding.

Rather than focusing on asymptotic quantum computational advantage over classical solvers, the present work introduces a new algorithmic principle for quantum time-marching methods. By showing that repeated non-unitary evolution can be implemented unconditionally through stochastic classical dynamics, we identify a framework that may guide the development of future quantum algorithms for nonlinear PDEs.

To place the proposed method within the current landscape of quantum PDE solvers, Table \ref{tab:comparison} compares the principal algorithmic approaches according to their numerical formulation, applicability to linear and nonlinear equations, implementation of non-unitary evolution, and resource requirements. Since the computational cost of quantum PDE solvers depends strongly on the adopted encoding, initialization procedure, and observable extraction, the reported complexities should be interpreted as representative scaling laws under the assumptions discussed in the corresponding references. 

\begin{table}[ht]
\centering
\begingroup
\renewcommand{\thefootnote}{\alph{footnote}}
 \begin{tabular}{||m{0.5cm} c c c | c | c | m{1.5cm}||} 
 \hline
 Ref.  & Qubits & Depth & Overall prob & Linear & Non-Lin & Eqs. \\ [0.5ex] 
 \hline\hline
 \cite{over2025quantum} & $O(\log N)$ & $\tilde{O}(TN^2/\epsilon)$ & $O( \frac{||\phi(0)||^2}{||\phi(T)||^2})$ & $\times$ &  & AD \\
 \hline
 \cite{brearley2024quantum} & $O(\log N)$ & $O(TNk/\epsilon)$ & $\approx 1$\footnotemark[1] & $\times$ & & Advection \\
 \hline
 \cite{bediche2025fully} & $O(\log N)$ & $O(T\log^2N)$ & $p$\footnotemark[2] & $\times$ & & AD (LBM) \\
 \hline
 \cite{wawrzyniak2025linearized} & $O(\log N)$ & $O(T\log ^2N)$ & 1 & $\times$ &  & AD (LBM) \\
 \hline
 \cite{esmaeilifar2024quantum} \cite{koukoutsis2025time} & $O(T)$ & $O(3^T)$ & $O(p^{3^T})$ & & $\times$ & Burgers Lorenz \\
 \hline
 \cite{fonio2025adaptive} & $\log (N)+T$ & $O(T\log(N))$ & $p^T$ & & $\times$ & Other \\
 \hline
 \cite{wang2025quantum} & $O(\log (N)+m)$ & $O(m^{T}\log N)$ & 1\footnotemark[3] & & $\times$ & LBM \\
 \hline
 \cite{schalkers2024importance} & $\Omega(mT),O(mN)$ & $O(T^2)$ & 1\footnotemark[3] & & $\times$ & LBM \\
 \hline
 \textbf{This work} & \textbf{$mN$} & \textbf{$O(T)$} & \textbf{1} & & $\times$  & \textbf{Burgers}\\ [1ex] 
 \hline
 \end{tabular}
 \renewcommand{\thefootnote}{\alph{footnote}}
 \footnotetext[a]{For the precise probability, refer to the original paper}
 \footnotetext[b]{A trash state is created, but with a probability independent on the number of timesteps}
 \footnotetext[c]{Nonlinearities introduce a probability of success with LCU or QSVT}
 \caption{Comparison of representative quantum algorithms for partial differential equations. Reported complexities correspond to the assumptions adopted in the respective references.}
 \label{tab:comparison}
\endgroup
\end{table}
The remainder of this article is organized as follows. Section~\ref{sec:frame} reviews the LCU framework and derives the conditions required for probabilistic non-unitary evolution. Section~\ref{sec:appli} applies this framework to the design of the collision step of a lattice gas cellular automaton, demonstrating that time-step concatenation is possible within the LCU formalism. The same framework is then applied to a probabilistic finite-difference method for the advection equation using an amplitude encoding compatible with potential quantum advantage. In this case, however, we show that the proposed encoding does not allow time-step concatenation.

\section{Methods}\label{sec:frame}

In the introduction, we argued that the main bottleneck of TM algorithms is repeated probabilistic application of non-unitary operators. Before constructing a Burgers solver, we first ask a more general question: which non-unitary operators can be realized by an LCU? The following framework answers this question independently of the particular PDE.
In the first place, we give the following definition
\begin{definition} [Probabilistically applied operators]
    A set of operators $\{O_m\}\in\mathbb{C}^{2^n\times 2^n}$ is probabilistically applied to a pure state of $n$ qubits $\ket{\psi}\in\mathcal{H}^{2^n}$ if the state evolves to
    \begin{equation}
        \ket{\psi'}=\frac{1}{\sqrt{p_m}}O_m\ket{\psi}
    \end{equation}
    with probability
    \begin{equation}
        p_m = \bra{\psi}O_m^{\dagger}O_m\ket{\psi}
    \end{equation}
    having
    $$\sum_m p_m=1,$$
    equivalently,
    $$\sum_m O_m^{\dagger}O_m=1$$
\end{definition}
We remark that the operators being probabilistically applied are in general non-unitary. In fact, with this definition we have that
\begin{equation}
    \sum_mO_m^{\dagger}O_m=I,
\end{equation}
and not the unitarity of $O_m$.Then, we give the following definitions
\begin{definition}
    A linear combination of unitaries (LCU) algorithm is a quantum algorithm involving $r$ ancillas and $n$ target qubits, carried out by the unitary operator $\hat{L}$ acting on $\mathcal{H}^{2^r}\otimes \mathcal{H}^{2^n}$ where
    \begin{equation}
        \hat{L} = \qty(\hat{H}_1 \otimes \hat{I}^{\otimes n}) \hat{W} \qty(\hat{H}_0 \otimes \hat{I}^{\otimes n})
    \end{equation}
    with
    $$\hat{W}=\sum_{s=0}^{2^r-1} \Pi_s \otimes \hat{U}_s$$
    having a set of $2^r$ projectors $\{\Pi_s\}$, a set of $2^r$ unitary operations $\{\hat{U}_s\}$ acting on the target qubits, a unitary operation $\hat{H}_0$ called "preparation" and a unitary operation $\hat{H}_1$ called "un-preparation" acting on ancillary qubits. The circuit representation is given in Fig.\ref{fig:LCU_4}
\end{definition}
\begin{figure}[ht]
    \centering
    \includegraphics[scale=1]{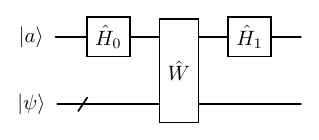}
    \caption{General LCU quantum circuit}
    \label{fig:LCU_4}
\end{figure}
\begin{definition}[LCU-conjugated operators]
    Let $\hat{L}$ be the unitary operator of a LCU acting on $\mathcal{H}^{anc}\otimes\mathcal{H}^{targs}$. Suppose we measure the ancillary register in the computational basis. For each possible ancilla measurement outcome $i$, we define the LCU-conjugated operator as 
    \begin{equation}
        C_i = \bra{i}\hat{L}\ket{0^{\otimes r}}_{anc}
    \end{equation}
\end{definition}
These are then the operators acting only on the target register, supposing a measurement occurred on the ancilla register. We can then prove the following lemma
\begin{lemma}
    Let $\hat{L}$ be the unitary operator of a LCU. Suppose we measure the ancillary register in the computational basis. The set of related LCU-conjugated operators $\{C_i \}$ is probabilistically applied to the target register.
\end{lemma}
This is a direct consequence of the measurement postulate of quantum mechanics \cite{nielsen2010quantum}.

We now consider a LCU with 1 ancilla. We represent the one-qubit preparation and unpreparation with unitaries parametrized as follows
\begin{equation}
    \hu(\theta,\zeta,\xi) = \begin{pmatrix}
        e^{-i(\zeta+\xi)}\cos{\theta}  & -e^{-i(\zeta-\xi)}\sin{\theta} \\
        e^{i(\zeta-\xi)}\sin{\theta} & e^{i(\zeta+\xi)}\cos{\theta}
    \end{pmatrix}
\end{equation}
where we neglect a global phase, considering $\theta,\zeta,\xi$ to be real angles. The specific quantum circuit of the algorithm is represented in Fig.\ref{fig:LCU_3}.
\begin{figure}[ht]
\centering
\includegraphics[scale=1]{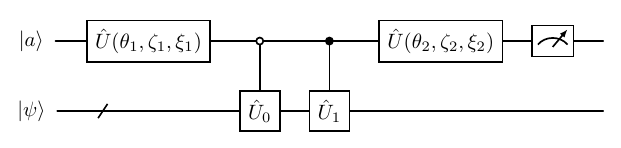}
\caption{Standard linear combination of unitaries.}
\label{fig:LCU_3}
\end{figure}
In this case, $\ket{\psi}$ is our target state, $\ket{a}$ is an ancilla supposed to be prepared in state $\ket{0}$, $\hat{U}(\theta_1,\zeta_1,\xi_1)$ is the unitary that prepares the superposition state of the ancilla, $\hat{U}_0$ and $\hat{U}_1$ are general unitary operations acting on the target qubits, $\hat{U}(\theta_2,\zeta_2,\xi_2)$ is the unitary that prepares the ancilla before the measurement. The final state before the measurement results in 
\begin{equation}
    \begin{split}
        \ket{\Psi} & = \ket{0}(e^{-i(\Delta_1^++\Delta_2^+)}\cos\theta_1\cos\theta_2 \hu_0 - e^{i(\Delta_1^--\Delta_2^-)}\sin\theta_1\sin\theta_2 \hu_1)\ket{\psi} + \\
        & + \ket{1}(e^{-i(\Delta_1^+-\Delta_2^-)}\cos\theta_1\sin\theta_2 \hu_0 + e^{i(\Delta_1^-+\Delta_2^+)}\sin\theta_1\cos\theta_2 \hu_1)\ket{\psi}
    \end{split}
\end{equation}
where $\Delta_i^{\pm}=\zeta_i \pm \xi_i$. In this way, the probabilities for a projective measurement of the ancilla, represented with projective operators  $M_i=\ket{i}\bra{i}\otimes \hat{I}^{\otimes n}$ with $i=0,1$, are
\begin{equation}
    p_m= \bra{\Psi}M_m^{\dagger}M_m\ket{\Psi},
\end{equation}
which are explicitly
\begin{equation} \label{eq:probs}
    \begin{split}
        p_0 & = c_1^2c_2^2 + s_1^2s_2^2 - c_1s_1c_2s_2 \bra{\psi}\hat{U}_{01} + \dagg{\hat{U}_{01}}\ket{\psi} \\
        p_1 & = c_1^2s_2^2 + s_1^2c_2^2 + c_1s_1c_2s_2 \bra{\psi}\hat{U}_{01} + \dagg{\hat{U}_{01}}\ket{\psi}
    \end{split}
\end{equation}
where $c_j=\cos(\theta_j)$, $s_j=\sin(\theta_j)$, and
\begin{equation}
    \hat{U}_{01} = e^{i2(\zeta_1+\xi_2)}\dagg{\hat{U}}_0\hat{U}_1
\end{equation}
The quantum state after the measurement will either be
\begin{equation} \label{eq: post-meas}
    \begin{split}
    \ket{\Psi_0} = \frac{1}{\sqrt{p_0}} \ket{0}(e^{-i(\Delta_1^++\Delta_2^+)}\cos\theta_1\cos\theta_2 \hu_0 - e^{i(\Delta_1^--\Delta_2^-)}\sin\theta_1\sin\theta_2 \hu_1)\ket{\psi}  \\
    \ket{\Psi_1} = \frac{1}{\sqrt{p_1}} \ket{1}(e^{-i(\Delta_1^+-\Delta_2^-)}\cos\theta_1\sin\theta_2 \hu_0 + e^{i(\Delta_1^-+\Delta_2^+)}\sin\theta_1\cos\theta_2 \hu_1)\ket{\psi}
    \end{split}
\end{equation}
depending on the outcome of the measurement. This allows us to clearly identify the non-unitary operations we are carrying out. These correspond to
\begin{align}
    A_0 & = \qty[ e^{-i(\Delta_1^++\Delta_2^+)}\cos\theta_1\cos\theta_2 \hu_0 - e^{i(\Delta_1^--\Delta_2^-)}\sin\theta_1\sin\theta_2 \hu_1 ] \label{eq:nu_condition_1}\\
    A_1 & = \qty[e^{-i(\Delta_1^+-\Delta_2^-)}\cos\theta_1\sin\theta_2 \hu_0 + e^{i(\Delta_1^-+\Delta_2^+)}\sin\theta_1\cos\theta_2 \hu_1] \label{eq:nu_condition_2}
\end{align}
Thus, we can say that if we measure the ancilla in the state $\ket{i}$, the state at the end of the LCU is
\begin{equation}
    \ket{\Psi'}=\frac{1}{\sqrt{p_i}} \ket{i} \otimes A_i \ket{\psi}
\end{equation}
The operators $A_0$ and $A_1$ are LCU-conjugated via the circuit in Fig.\ref{fig:LCU_3}.
These calculations can then be summarized in the following lemma
\begin{lemma}
    The LCU-conjugated operators via the circuit in Fig.\ref{fig:LCU_3} result in
    \begin{align}
        A_0 & = \qty[ e^{-i(\Delta_1^++\Delta_2^+)}\cos\theta_1\cos\theta_2 \hu_0 - e^{i(\Delta_1^--\Delta_2^-)}\sin\theta_1\sin\theta_2 \hu_1 ] \\
    A_1 & = \qty[e^{-i(\Delta_1^+-\Delta_2^-)}\cos\theta_1\sin\theta_2 \hu_0 + e^{i(\Delta_1^-+\Delta_2^+)}\sin\theta_1\cos\theta_2 \hu_1] 
    \end{align}
    and form a probabilistically applied set of operators, with respective probabilities
    \begin{align*}
        p_0 & = c_1^2c_2^2 + s_1^2s_2^2 - c_1s_1c_2s_2 \bra{\psi}\hat{U}_{01} + \dagg{\hat{U}_{01}}\ket{\psi} \\
        p_1 & = c_1^2s_2^2 + s_1^2c_2^2 + c_1s_1c_2s_2 \bra{\psi}\hat{U}_{01} + \dagg{\hat{U}_{01}}\ket{\psi}
    \end{align*}
    where $c_j=\cos(\theta_j)$, $s_j=\sin(\theta_j)$, and
    \begin{equation*}
        \hat{U}_{01} = e^{i2(\zeta_1+\xi_2)}\dagg{\hat{U}}_0\hat{U}_1
    \end{equation*}
\end{lemma}
The proof is given in the above calculations. We remark that $A_0$ and $A_1$ are, in general, non-unitary. The core point is: the process of measuring induces a probabilistic evolution of the target state, and this evolution can be described in terms of a non-unitary evolution. We will leverage this property to define a quantum version of a classical algorithm where the evolution is non-unitary and probabilistic, and thus showing that it can be translated into LCU-conjugated operators. The fact of knowing the measurement outcome allows us to have a resulting pure state, enabling us to apply the evolution again, which is the fundamental feature we are looking for in TM algorithms.

We started by knowing the unitary operations involved in a LCU and derived the form of the non-unitary probabilistic operations we apply through the algorithm in Fig.\ref{fig:LCU_3}, supposing we know the measurement outcome. We can also go the other way around and ask: given two non-unitary operations, can they be LCU-conjugated? Suppose we have two non-unitary operations $\{E_0,E_1\}\in\mathbb{C}^{2^n \times 2^n}$, and we want to find out if we can carry them out probabilistically with a LCU. This means to check if they can be written as LCU-conjugated operations. Thus, we give the following theorem
\begin{theorem} \label{th:1}
    Given two non-unitary operators $E_0$ and $E_1$, they are LCU-conjugated via the circuit in Fig. \ref{fig:LCU_3} if and only if
    \begin{equation}\label{eq:completeness}
        \sum_i E^{\dagger}_i E_i = \hat{I}
    \end{equation}
    and $\exists \theta_1\neq k\pi/2$ with $k\in\mathbb{Z}$, and $\exists\zeta_2,\theta_2\in \mathbb{R}$ such that
    \begin{equation}\label{eq:second_condition}
        (c_2^2-s_2^2)(E_0^{\dagger}E_0-E_1^{\dagger}E_1)+2s_2c_2(e^{-i2\zeta_2}E_0^{\dagger}E_1 + e^{i2\zeta_2}E_1^{\dagger}E_0) = (c_1^2-s_1^2) \hat{I}
    \end{equation}
    with $c_j = \cos\theta_j$ and $s_j = \sin\theta_j$.
\end{theorem}

\begin{proof}
We start by considering $\theta_1\neq k\pi/2$ with $k\in\mathbb{Z}$, and $\xi_i,\zeta_i,\theta_2\in \mathbb{R}$, and we define
\begin{align}
    V_0 & = \frac{e^{i(\xi_1+\xi_2+\zeta_1)}}{\cos{\theta_1}} [e^{i\zeta_2}\cos{\theta_2}E_0 + 
    e^{-i\zeta_2}\sin{\theta_2}E_1 ] \label{eq:U0}\\
    V_1 & = -\frac{e^{i(\xi_1-\xi_2-\zeta_1)}}{\sin{\theta_1}} [e^{i\zeta_2}\sin{\theta_2} E_0 - e^{-i\zeta_2}\cos{\theta_2} E_1] \label{eq:U1}
\end{align}
These are, in general, non-unitary. We can now impose the unitarity of $V_0$ and $V_1$, thus asking $V_0^{\dagger}V_0=I$ and $V_1^{\dagger}V_1=I$. These two constraints can be explicitly calculated, and they result in Eqs.\ref{eq:completeness},\ref{eq:second_condition}. Thus, we can conclude that if $E_0$ and $E_1$ satisfy Eqs.\ref{eq:completeness}\ref{eq:second_condition}, then $V_0$ and $V_1$ in Eqs.\ref{eq:U0},\ref{eq:U1} are unitary. We can now express $E_0$ and $E_1$ in terms of $V_0$ and $V_1$. In doing so, we obtain the expressions in Eqs.\ref{eq:nu_condition_1},\ref{eq:nu_condition_2} with $\hat{U}_0=V_0$ given in Eq.\ref{eq:U0}, and $\hat{U}_1=V_1$ given in Eq.\ref{eq:U1}. These expressions were coming directly from the algorithm in Fig.\ref{fig:LCU_3}, and are then LCU-conjugated. Thus, we can conclude that if two non-unitary operations respect Eqs.\ref{eq:completeness}, \ref{eq:second_condition}, then they are LCU-conjugated via the circuit in Fig.\ref{fig:LCU_3}.
\end{proof}
The necessity is confirmed by the above calculations. In other words, if $C_0$ and $C_1$ satisfy the desired conditions, they can be probabilistically applied to a target quantum state via a LCU circuit. In the specific case of an H-LCU (where Hadamard gates are applied to the ancilla before and after the controlled operations), the condition Eq.\ref{eq:second_condition} simplifies to:
\begin{equation} \label{eq:pseudo_H}
    C_0^{\dagger}C_1 + C_1^{\dagger}C_0 = 0
\end{equation}
This motivates the following definition:

\begin{definition}
    Two square matrices $A, B \in \mathbb{C}^{n \times n}$ are \textbf{pseudo-commutative} if $A^\dagger B + B^\dagger A = 0$.
\end{definition}

Thus, in an H-LCU, the non-unitary operators being probabilistically applied must be pseudo-commutative and form a complete set. We refer to Eq. \ref{eq:second_condition} as the \textbf{conditional pseudo-commutativity condition}. We remark that while completeness depends only on the LCU-conjugated operations, conditional pseudo-commutativity depends also on the specific ancillary rotations. 

Our claim then is that if two operators are LCU-conjugated, then they can be part of a non-unitary evolution in a TM algorithm. In fact, after the measurement, it is possible to apply the same operation again, keeping in mind that the quantum state will be one of Eqs.\ref{eq: post-meas}. If two operators cannot be LCU-conjugated, then they do not satisfy Theorem\ref{th:1}, and if two operators do not satisfy Theorem\ref{th:1}, then they cannot be LCU-conjugated via the circuit in Fig.\ref{fig:LCU_3}. It is still possible to carry them out with a LCU, but this will generally produce an undesired state, which is an unfavourable feature for TM algorithms, and needs to be treated with amplitude amplification strategies.

We point out that the setting we used is not the only possible one. In fact, we could generalize to POVM measurements, for example. Furthermore, Theorem \ref{th:1} is specific to the circuit in Fig. \ref{fig:LCU_3} and to measuring in the computational basis, which can further be generalized. In the following, we see how the properties we derived can be used to design the collision operator of a quantum lattice gas automata for Burger's equation.

\section{Results}\label{sec:appli}

The previous section established general conditions under which non-unitary operators can be probabilistically applied to a target qubit register. We now investigate whether these conditions are satisfied by existing probabilistic numerical algorithms. We first consider a lattice gas cellular automaton for Burgers' equation, which provides a positive example, before later examining a randomly sampled finite-difference discretization for which this framework allows us to assess the limitations of the encoding adopted.

\subsection{Quantum lattice gas algorithm for Burgers' equation}

Lattice gas cellular automata (LGCA) are among the classical numerical methods capable of simulating PDEs that are natively probabilistic. In particular, there is one LGCA defined in \cite{boghosian2019cellular} which is capable of simulating 1D Burgers' equation with a set of operators that can be LCU-conjugated, if we allow a slight modification of the relative phase. This modification is consistent, however, with the encoding and with the consequent retrieval of information.

The LGCA proposed in \cite{boghosian2019cellular} is characterized as follows. We consider a 1D lattice consisting of $N$ sites. In each site, there are two bits $[b_-,b_+]$ representing the presence of a right-moving (left-moving) particle with $b_+$ ($b_-$) in the corresponding site. An exclusion principle holds; thus, in each site there can be at most one particle per velocity. This system undergoes an evolution consisting of a collision process and a streaming process. The collision process changes the state of each cell synchronously. Its truth table is Table \ref{table:1}.

\begin{table}[h!]
\centering
\begin{tabular}{c|c|c|cc}
$b_-(x,t)$ & $b_+(x,t)$ & $b_-'(x,t)$ & $b_+'(x,t)$   \\ \hline 
0    & 0    & 0  & 0  \\
0    & 1    & $\frac{1-\alpha}{2}$ & $\frac{1+\alpha}{2}$   \\
1    & 0    & $\frac{1-\alpha}{2}$ & $\frac{1+\alpha}{2}$   \\
1    & 1    & 1 & 1 
\end{tabular}
\caption{Truth table of the collision process. $b_i(x,t)$ is the precollision state at time $t$ of the cell at site $x$, $b_i'(x,t)$ is the postcollision state. $\alpha$ is a random variable valued 1 or -1.}
\label{table:1}
\end{table}

In this process, $\alpha$ can be either 1 with probability $p$ or -1 with probability $1-p$. 
As we can see from Table \ref{table:1}, the process is not reversible, thus cannot be represented with a unitary operator. The streaming step consists of the particles moving according to their velocity. An example of evolution is represented in Fig.\ref{fig:boghosian}
\begin{figure}[ht]
    \centering
    \includegraphics[scale=0.85]{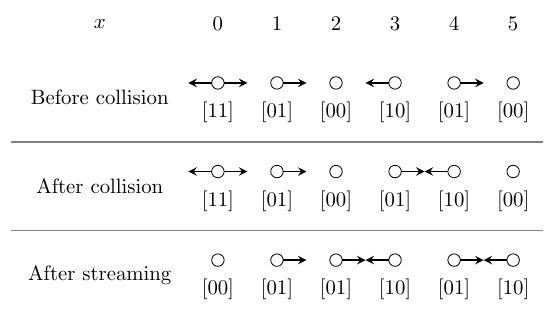}
    \caption{Example of classical evolution \cite{boghosian2019cellular}. The collision takes place  at $x=1,3,4$, with different random extractions for $\alpha$. All the other cells are not affected by the collision. The streaming takes place according to respective velocities with continuous boundary conditions.}
    \label{fig:boghosian}
\end{figure}
It is proven in \cite{boghosian2019cellular} that this system collectively behaves according to Burgers' equation. We provide in Fig.\ref{fig:sims} some numerical results for visualizing it

\begin{figure}
    \centering
    \includegraphics[width=0.48\linewidth]{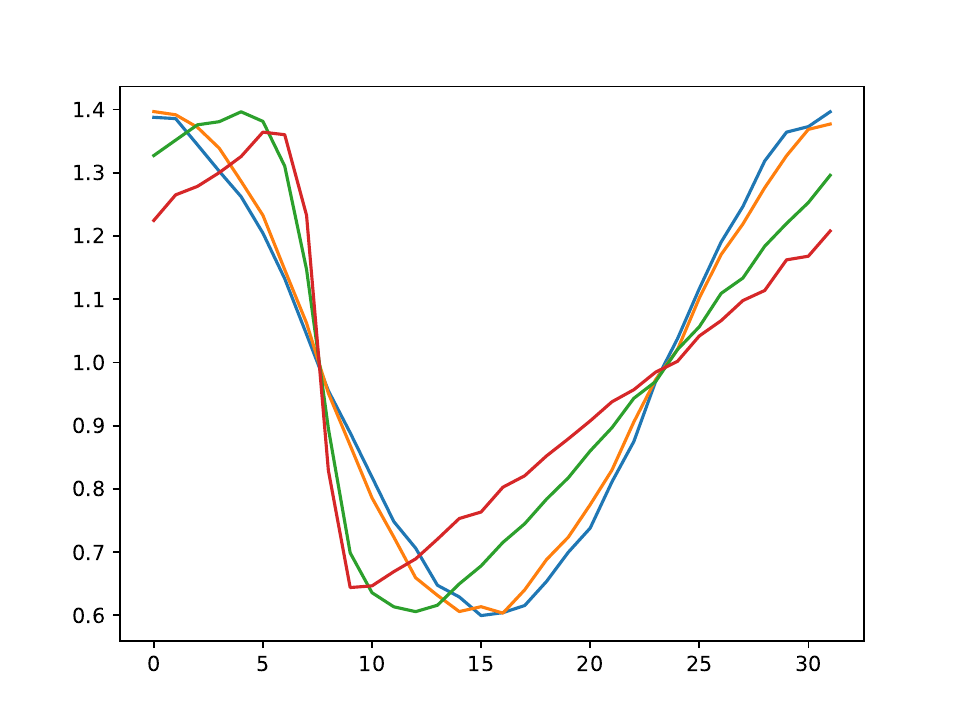}
    \includegraphics[width=0.48\linewidth]{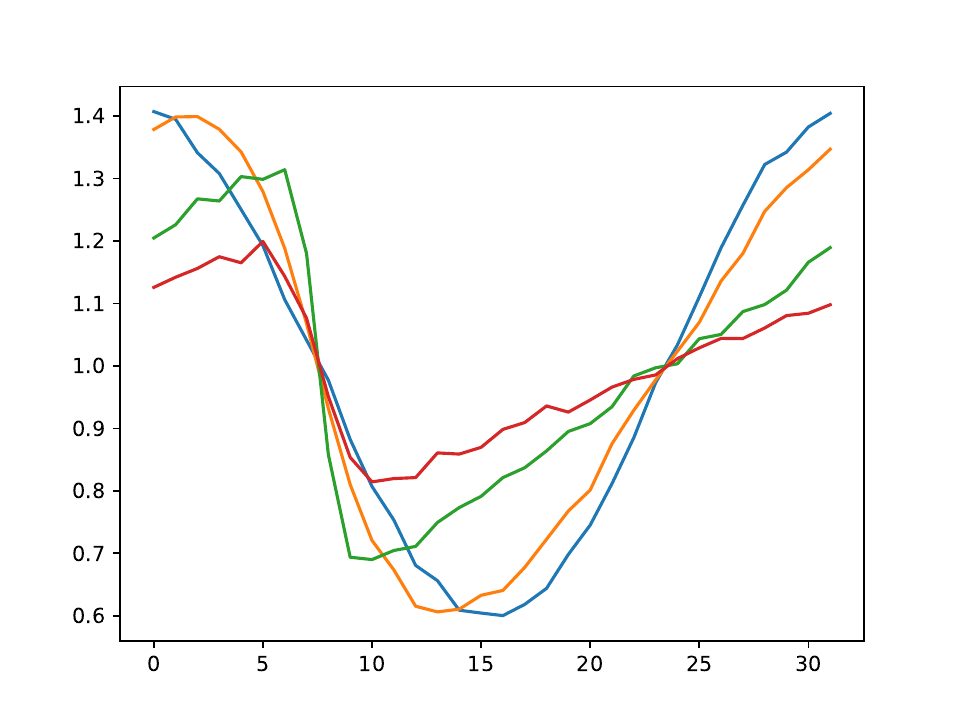}
    \caption{Simulation of Burger's equation with LGCA. 2048 grid points were used, ensemble-averaged over 100 realizations and space-averaged over 64 grid points. The left plot has $p=0.75$, the right plot has $p=0.5$. The specific dependence of the hydrodynamic parameters is partly detailed in \cite{boghosian2019cellular}.}
    \label{fig:sims}
\end{figure}


We now give the quantum version of the algorithm. First, we consider the computational basis encoding for our algorithm, thus having 
\begin{equation}
    [b_-b_+]\leftrightarrow\ket{b_-b_+}.
\end{equation} 
The encoding of the occupation state is crucial for the design of the collision operation. The overall complexity is dominated by the total encoding of the lattice, which is discussed later.



For the collision we need to apply the operators $C_0$ and $C_1$ with probabilities $p$ and $1-p$, being

\begin{align}
    C_0=\begin{pmatrix}
        1 & 0 & 0 & 0 \\
        0 & 1 & 1 & 0 \\
        0 & 0 & 0 & 0 \\
        0 & 0 & 0 & 1
    \end{pmatrix}
     & &
     C_1=\begin{pmatrix}
        1 & 0 & 0 & 0 \\
        0 & 0 & 0 & 0 \\
        0 & 1 & 1 & 0 \\
        0 & 0 & 0 & 1
    \end{pmatrix}
\end{align}
We consider the first case where $p=1/2$. These operators are clearly non-unitary, and they do not respect the LCU-conjugated condition for H-LCU. In fact, we can calculate directly that the completeness condition is not respected.
However, we can do a slight modification and consider a relative phase in the outgoing states. In particular, we multiply by a different phase $e^{i\alpha_j}$ the $j$-th column of $C_0$,$e^{i\beta_j}$ the $j$-th column of $C_1$. We can do this because, for the LGCA, the phase is not important. With the computational basis encoding adopted, the information is in the quantum state, while the phase is relatively unimportant, since a measurement of the cell would bring the same correct result anyway. In this sense, the encoding is phase-independent. From the completeness condition Eq. \ref{eq:completeness}, we have
\begin{equation} \label{eq:rel_phases_cond_comp}
    \alpha_1+\beta_2=\alpha_2+\beta_1+\pi
\end{equation}
and from the pseudo-commutativity in Eq. \ref{eq:pseudo_H} we have
\begin{equation}\label{eq:rel_phases_cond_pseudo}
    \begin{split}
        \alpha_0 &= \beta_0+\pi/2 \\
        \alpha_3 &= \beta_3+\pi/2
    \end{split}
\end{equation}
Then, a possible choice of parameters satisfying these two conditions gives the following operators
\begin{align}
    C_0'=\begin{pmatrix}
        1 & 0 & 0 & 0 \\
        0 & 1 & 1 & 0 \\
        0 & 0 & 0 & 0 \\
        0 & 0 & 0 & 1
    \end{pmatrix}
     & &
    C_1'=\begin{pmatrix}
        i & 0 & 0 & 0 \\
        0 & 0 & 0 & 0 \\
        0 & 1 & -1 & 0 \\
        0 & 0 & 0 & i
    \end{pmatrix}
\end{align}
These are H-LCU-conjugated, and in particular we can find $\hat{U}_0$ and $\hat{U}_1$ resulting in 
\begin{align}
    \hat{U}_0=\frac{1}{\sqrt{2}}\begin{pmatrix}
        1+i & 0 & 0 & 0 \\
        0 & 1 & 1 & 0 \\
        0 & -1 & 1 & 0 \\
        0 & 0 & 0 & 1+i
    \end{pmatrix}
     & 
     \hat{U}_1=\frac{1}{\sqrt{2}}\begin{pmatrix}
        1-i & 0 & 0 & 0 \\
        0 & 1 & 1 & 0 \\
        0 & 1 & -1 & 0 \\
        0 & 0 & 0 & 1-i
    \end{pmatrix}
\end{align}
This proves that this operation can be carried out with the quantum circuit in Fig.\ref{fig:LCU_3} in a TM fashion, without any postprocessing and with an unconditional certainty of obtaining the desired state. Since the algorithm is probabilistic, at each time step we probabilistically apply one of the operators $C_i'$, guaranteeing the correct evolution of the state.

H-LCU allows, however, for applying these operations with probability $1/2$. To apply them with arbitrary probabilities, we need to satisfy Eq. \ref{eq:second_condition} instead of the simplified Eq. \ref{eq:pseudo_H} for using Hadamard gates. The completeness condition holds in any case because it is independent of the operation on the ancilla; thus we still have Eq. \ref{eq:rel_phases_cond_comp}. For Eq. \ref{eq:second_condition} we have in the first place that $\cos(2\theta_1)=0$, thus $\theta_1=\pi/4$ up to a periodic term. This condition, however, brings the probabilities $p_i$ of Eq. \ref{eq:probs} to be the same in the simplified case of pseudo-commutative unitaries. Thus, we conclude it is not possible to carry out the evolution with different probabilities.

In this section, we showed how LCU followed by a measurement can be used to implement, with unconditional success probability, a probabilistic algorithm. The properties previously derived allowed us to design ad hoc operators and adapt the prescribed encoding, making it possible to carry out the collision step of a LGCA, avoiding postprocessing. 

\subsection*{Complexity}
If we want to consider a 1D lattice consisting of $N$ cells, then we can have different encodings of the entire lattice, always using computational basis encoding. For simplicity, we consider having two qubits in each cell $\ket{q_-}\otimes\ket{q_+}$. The total lattice state is 
\begin{equation}\label{eq:encoding_lattice}
    \ket{\Psi(t)}=\bigotimes_x\ket{q_-(x,t)q_+(x,t)}.
\end{equation} 

Other encodings, such as space-time encoding, can be considered \cite{georgescu2025fully}. With the encoding as in Eq.\ref{eq:encoding_lattice}, we carry out the collision in each site independently and in parallel, with a complexity of $O(1)$, while the streaming corresponds to a series of SWAPs, also applicable in parallel, resulting in $O(2)$. This means an asymptotic time complexity of $O(T)$ for simulating $T$ steps. With this encoding, however, we need at least 2 qubits per site, thus having a space complexity of $\Omega(2N)$. Considering the worst-case scenario of one ancilla per site, the space complexity is $O(3N)$. This configuration is, generally, non-advantageous.

We remark, however, that the proposed procedure works with any quantum algorithm that uses computational basis encoding; thus, its effectiveness is independent of the encoding of the entire lattice. As a future perspective, the study of fundamental limitations to an advantageous encoding of the entire lattice could allow for a quantum advantage \cite{fonio2025quantum}, thus making our procedure a fundamental step for quantum nonlinear PDEs solvers.

In the following, we analyze the possibilities of a probabilistic version of Euler's scheme for the advection equation.



\subsection{Finite difference method (FDM)}

A quantum implementation of FDM is at the core of several quantum algorithms. There are different possible strategies to implement it. For example, the differential operator can be cast into an anti-Hermitian operator to give rise to a unitary dynamics for convection problems \cite{zylberman2026trotter}. Alternatively, for TM algorithms, a quantum representation of FDM was given in \cite{over2025quantum}, where they achieved the success probability independent of the number of time steps for the advection-diffusion equation. Seeing these alternatives, we could ask if LCU-conjugated operations can be useful for implementing a FDM scheme. We give in the following a first naive scheme, proving that the properties we obtained show how it is not possible to implement a random first-order Euler scheme for the advection equation, thus proving their possible use as design principles of stochastic quantum algorithms.

We consider the 1D example of a variable $f(x,t)$ which can be represented in a vector $\mathbf{f}(t)=[f(0,t),f(1,t),\cdots,f(N,t)]$ if we consider $N=2^n$ gridpoints. We consider an amplitude encoding

\begin{equation} \label{eq:encoding_fdm}
    \ket{\psi(t)}=\alpha \sum_x f(x,t)\ket{x}
\end{equation}
This is a very common encoding, where the lattice is represented in a superposition state, whose coefficients correspond to the normalized variable values. 

We consider the simulation of a simple advective PDE such as

\begin{equation} \label{eq:advection}
    \partial_tf(x,t)=\alpha\partial_x f(x,t)
\end{equation}
Solving this equation with FDM means writing a differential operator $\mathbf{D}$ for the partial derivative with respect to $x$, and computing $\mathbf{f}(t+1) = (\mathbf{I}+\mathbf{D})\mathbf{f}(t)$. We can use for this purpose the backward operator $\mathbf{D_b}$, the forward operator $\mathbf{D_f}$, or the central operator $\mathbf{D_c}$. Their form depends on the boundary conditions: we verified the following results for periodic BC, when they are defined as follows. 
\begin{equation}
\begin{split}
& \mathbf{D_f}= \frac{1}{h}
\begin{bmatrix}
-1 & 1 & \cdots & 0 \\
0 & -1 & \cdots & 0 \\
\vdots & \ddots & \ddots & \vdots \\
1 & 0  & 0 & -1
\end{bmatrix}
\mathbf{D_b} = \frac{1}{h}
\begin{bmatrix}
-1 & 0 & \cdots & 1 \\
1 & -1 & \cdots & 0 \\
\vdots & \ddots & \ddots & \vdots \\
0 & \cdots & 1 & -1
\end{bmatrix} \\
&\mathbf{D_c} = \frac{1}{2}(\mathbf{D_f}+\mathbf{D_b})
\end{split}
\end{equation}
These differential operators are non-unitary; thus, if we want to calculate the updated $\ket{\psi(t+1)}$, we need to carry out a non-unitary operation. This can be done, as we said, with a block-encoding, which, however, is not feasible for multi-time-step simulations. At this point, we can check if we can formulate the TM evolution with LCU-conjugated operators. In this sense, we could consider applying $C_0=\mathbf{D}_i$ and $C_1=\mathbf{D}_j$ with $i\neq j$ for different ancillary states. In this way, we could calculate the space derivative in any case, but probabilistically applying a non-unitary operation.

It is straightforward, thanks to the conditions we derived in Eqs.\ref{eq:completeness},\ref{eq:second_condition}, to realize that no combination of $\mathbf{D}_i$ and $\mathbf{D}_j$ is LCU-conjugated. This means that we cannot evolve the amplitudes of Eq.\ref{eq:encoding_fdm} such that they reproduce the prescribed classical FDM scheme in a TM quantum computation using a simple LCU. This shows how useful the conditions of Theorem \ref{th:1} are to assess the validity of a possible encoding for TM algorithms. 

At this point, we could try to modify the non-unitary operators as in the previous case. However, this is not straightforward since the chosen encoding for FDM is not phase-independent. A relative phase on one of the terms of the finite difference operation changes the derivative we are applying.

We notice that FDM is not natively probabilistic, but we were still able to define it in a probabilistic way to try to check for LCU-conjugated operations. For quantum lattice Boltzmann methods (QLBM), this is actually not straightforward, and we leave it as a future perspective.

\section{Discussion}

In this article, we introduced LCU-conjugated operations, a class of non-unitary operations that are probabilistically applied to a target quantum state through an LCU procedure followed by measurement. We derived two fundamental properties using a single ancilla qubit, thereby characterizing the class of non-unitary operations that can be implemented in this way and providing a practical framework for the design of new quantum algorithms.

As a first application, we developed the quantum counterpart of a classical lattice gas cellular automaton for Burgers' equation. This demonstrates that LCU-conjugated operations enable the construction of a quantum algorithm with unconditional success probability, capable of simulating nonlinear dynamics without requiring multiple copies of the quantum state, unlike existing approaches. At the same time, we showed that an efficient quantum implementation remains limited by the lack of an efficient encoding for LGCA, which therefore remains an open problem.

As a second application, we showed that a straightforward probabilistic version of Euler's finite-difference method cannot be implemented using LCU-conjugated operations, demonstrating the need for an alternative formulation. More broadly, this example illustrates how the properties derived in this work provide a systematic tool for assessing the feasibility of probabilistic quantum algorithms. In this sense, LCU-conjugated operations may serve as fundamental building blocks in the design of future quantum algorithms, particularly for linear and nonlinear partial differential equations and other time-marching schemes.

The most promising perspective of this work lies in quantum algorithm design. In particular, the results presented here identify features of the LCU framework that can be directly exploited in the construction of quantum time-marching algorithms. This opens a possible route toward quantum algorithms for nonlinear PDEs and suggests a broader role for probabilistic quantum algorithms in this context. More generally, the proposed framework could be applied to any probabilistic time-marching formulation, including Monte Carlo methods, other lattice gas cellular automata, and classical probabilistic algorithms for nonlinear partial differential equations.

\backmatter

\section*{Acknowledgements}
This work was partly supported by the PEPR EPiQ ANR-22-PETQ-0007, and ANR JCJC DisQC ANR-22-CE47-0002-01

\section*{Statements and Declarations}

\noindent\textbf{Competing interests.}
The authors declare that they have no competing interests.

\bibliography{sn-bibliography}

\end{document}